\documentclass[a4paper,USenglish]{llncs}

\usepackage[utf8]{inputenc}
\usepackage[T1]{fontenc}
\usepackage{mathtools}
\usepackage{amsmath}
\usepackage{amssymb}
\usepackage[final]{microtype}

\usepackage[final,pdftex]{hyperref}

\usepackage{babel}
\usepackage{fixme,datetime,xspace,cleveref}
\usepackage{mytikz}
\usetikzlibrary{decorations.pathmorphing,fadings}

\input{define.orgtex}

\FXRegisterAuthor{mc}{amc}{MC}
\FXRegisterAuthor{sa}{sac}{SA}
\FXRegisterAuthor{gp}{gpc}{GP}
\FXRegisterAuthor{fm}{fmc}{FM}

\revision$LastChangedRevision: 1374 $

\title{Weak Cost Register Automata are Still Powerful}
\def\lastandname{\unskip,}
\author{Shaull Almagor\inst{1} \and Michaël Cadilhac\inst{1} \and Filip
  Mazowiecki\inst{2} \and Guillermo A. Pérez\inst{3}}
\institute{University of Oxford
  \and
  LABRI, Université de Bordeaux
  \and
  Université libre de Bruxelles}

\begin{document}

\maketitle

\begin{abstract}
  We consider one of the weakest variants of cost register automata over a
  tropical semiring, namely copyless cost register automata over \(\bbN\) with
  updates using \(\min\) and increments.  We show that this model can simulate, in
  some sense, the runs of counter machines with zero-tests.  We deduce that a
  number of problems pertaining to that model are undecidable, in particular
  equivalence, disproving a conjecture of Alur~et~al.~from 2012.  To emphasize
  how weak these machines are, we also show that they can be expressed as a
  restricted form of linearly-ambiguous weighted automata.
\end{abstract}

\vspace{2em}

\centerline{\textsc{Contents}}

\vspace{1em}

\tableofcontentsinline

\vfill

\newpage

\section{Introduction}

Cost register automata (CRA)~\cite{alur-et-al13} encompass a wealth of computation
models for functions from words to values (herein, integers).  In their full
generality, a CRA is simply a DFA equipped with registers that are updated upon
taking transitions.  The updates are expressions built using a prescribed set of
operations (e.g., \(+, \times, \min, \ldots\)), constants, and the registers themselves.

In this work, we will focus on CRA computing integer values, where the updates
may only use ``\(+c\)'', for any constant \(c\), and \(\min\).  For instance:
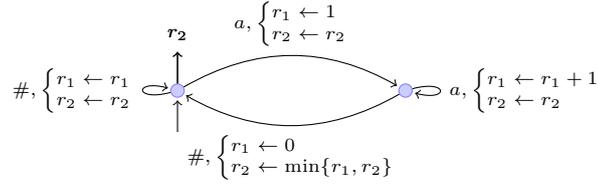
\begin{center}
  \begin{smallautomaton}[font={\scriptsize}]
    \node[smallstate,initial below,accepting by arrow, accepting above, accepting text=\(r_2\)] (p) {};
    \node[smallstate] (q) at (3, 0) {};
    \path[->] (p) [loop left] edge node {\(\#,
      \begin{update}
        r_1 & r_1\\
        r_2 & r_2
      \end{update}
      \)} ();
    \path[->] (p) [bend left] edge node {\(a,
      \begin{update}
        r_1 & 1\\
        r_2 & r_2\\
      \end{update}
      \)} (q)
    (q) [bend left] edge node {\(\#,
      \begin{update}
        r_1 & 0\\
        r_2 & \min\{r_1, r_2\}
      \end{update}
      \)} (p);
    \path [->] (q) [loop right] edge node {\(a,
      \begin{update}
        r_1 & r_1 + 1\\
        r_2 & r_2
      \end{update}
      \)} (q);
  \end{smallautomaton}
\end{center}
With \(r_1\) initialized to \(0\) and \(r_2\) to \(\infty\), this CRA computes the length of
the minimal nonempty block of \(a\)'s between two \(\#\)'s.  This model has the same
expressive power as weighted automata (WA) over the structure \((\bbZ, \min, +)\),
but the use of registers can simplify the design of functions.

The example above enjoys an extra property that can be used to restrain the
model (since a lot of interesting problems are undecidable on
WA~\cite{almagor-boker-kupferman11}).  Indeed, \emph{no register is used twice
  in any update function}; this property is called \emph{copylessness}.  This
syntactic restriction, introduced by Alur~et~al.~\cite{alur-et-al13} and studied
by Mazowiecki and Riveros~\cite{mazowiecki-riveros16}, provably weakens the
model.  It was the hope of Alur~et~al.~that this would provide a model for which
equivalence is decidable.

\paragraph{Semilinearity and decidability of equivalence.} Recall that
a set \(R \subseteq \bbZ^k\) is semilinear if it is expressible in first-order logic with
addition: \(\FO[<, +]\).  This latter logic being decidable~\cite{presburger27},
semilinearity is a useful tool to show decidability results.  For instance, let
\(f, g\colon A^* \to \bbZ\) be expressible in some model for which the images of
functions are effectively semilinear.  Suppose further that the function
\(h\colon w \mapsto \min\{2\times f(w), 2\times g(w) +1\}\) is also in that model.  Since the
image \(h(A^*)\) is effectively semilinear, one can check whether it is always
even: this would show that \(f(w) \leq g(w)\) for all \(w\).  A first natural question
is thus, is copyless CRA (CCRA) such a model?

\paragraph{Iterating \(\min\) breaks semilinearity.}  Deterministic automata
equipped with copyless registers with only ``\(+c\)'' updates are quite
well-behaved~\cite[Section~6]{cadilhac-finkel-mckenzie13}; in particular, the set
\[R = \{\vec{r} \mid \vec{r} \text{ are the values of the registers at the end of an
    accepting run}\}\]%
is semilinear.  Naturally, \(\min\{x, y\}\) is expressible in \(\FO[<, +]\), hence
\(\FO[<, +] = \FO[<, +, \min]\) (even, and this is not immediate, when the extra
value \(\infty\) is added~\cite{gaubert-katz04}).  This entails that if we were to give
to these automata the ability to do a \emph{constant} number of \(\min\), we would
still have that \(R\) is semilinear.  In this paper, it is shown that if the
number of \(\min\) is unbounded along runs, then the set is not semilinear (see
the proof of \Cref{thm:sl} for a simple construction), and that it is
undecidable to check whether \(R\) is semilinear.

\paragraph{Contributions.} Beyond considerations on semilinearity, we show that
CCRA over~\(\bbN\) can simulate the runs of counter machines with zero-tests
(\Cref{thm:sim}).  Intuitively, the only words mapped by the CCRA to an even
value are the correct executions of the counter machine.  This construction is
then used to show that equivalence is undecidable for CCRA over \bbN and that
upper-boundedness is undecidable for WA.  To better gauge the expressiveness of
CCRA, we show that they are a weak form of \emph{linearly-ambiguous} WA, that
is, WA for which no word~\(w\) has more than \(k \times |w|\) accepting runs, for some
constant \(k\) (see drawing on page~\pageref{fig:nf}).  Since the problems we
tackle are decidable for \emph{finitely-ambiguous} WA, CCRA are arguably the
simplest generalization of deterministic WA for which equivalence is
undecidable.

\section{Preliminaries}

We assume familiarity with automata theory, for which we settle some notations.
We write \(\bbN\) for \(\{0, 1, 2, \ldots\}\), \(\bbZ\) for the integers, and define
\(\bbNi = \bbN \cup \{\infty\}\) and \(\bbZi = \bbZ \cup \{\infty\}\).  Naturally,
\(\min\{\ldots, \infty, \ldots\}\) stays the same when removing the~\(\infty\) value, and we set
\(\min\emptyset = \infty\).  For any \(k \geq 1\), we write \([k]\) for \(\{1, 2, \ldots, k\}\).  We write
\(\eps\) for the empty word.

\paragraph{Automata.}  An automaton (NFA) is a tuple \((Q, A, \delta, q_0, F)\), where
\(Q\) is the set of states, \(A\) the alphabet,
\(\delta\subseteq Q \times (A \cup \{\eps\}) \times Q\) the transition relation, \(q_0\) the initial state, and
\(F \subseteq Q\) the set of final states.  We rely on the usual vocabulary pertaining to
automata: a |run| is a word in \(\delta^*\) starting in \(q_0\), and such that each
transition is consistent with the next; it is accepting if the last reached
state is in \(F\); a word \(w \in A^*\) is accepted if there is an accepting run
labeled by \(w\).

If \(\delta\) is a function from \(Q \times A\) to \(Q\), the automaton is |deterministic|
(DFA).  If there is a \(k \in \bbN\) such that each accepted word \(w\) is the label
of at most \(k \times |w|\) accepting runs, the automaton is |linearly-ambiguous|.

\paragraph{Tropicalities.} The only semirings (\ie, algebraic structures) that we
will use are \((\bbZi, \min, +)\) and \((\bbNi, \min, +)\), often dubbed ``tropical
semirings.''  When the discussion is not specific to one of the two semirings,
we simply write \(\bbK\) for both.  As with rings, matrix multiplication is
well-defined in semirings; e.g., if \((b_{ij})\) and \((c_{ij})\) are
\(2\times2\) matrices and \((a_{ij}) = (b_{ij})\cdot(c_{ij})\), then:
\[a_{2, 1} = \min \{b_{2,1}+c_{1,1}, b_{2,2} +
c_{2,1}\}\enspace.\]

\paragraph{Weighted automata.}  Weighted automata will only be used in
Section~\ref{sec:wa} and \Cref{thm:ub}.  A weighted automaton \(\cW\) over
\(\bbK\) (\kwa) is a tuple \((\cA, \lambda, \mu, \nu)\) where
\(\cA = (Q, A, \delta, q_0, F)\) is an NFA, and
\(\lambda \in \bbK, \mu\colon \delta \to \bbK,\) and \(\nu \colon F \to \bbK\).  Given a run
\(t_1\cdot t_2\cdots t_n \in \delta^*\) ending in a state \(q \in F\) in \(cA\), its \emph{weight} is
\(\lambda + \mu(t_1) + \mu(t_2) + \cdots + \mu(t_n) + \nu(q)\).  The weight \(\cW(w)\) of a word
\(w \in A^*\) is the minimum weight for all accepting runs over \(w\) in the NFA
(hence it is \(\infty\) if the word is not accepted).  The \kwa is |deterministic|
(\resp~|linearly-ambiguous|) if \(\cA\) is.  We use \(\kdetwa\) and \(\klinwa\) for
these restrictions.

\paragraph{Registers and counters.} A central goal of this work is to present a
simulation of some \emph{counter} machine with zero-tests by a \emph{register}
machine without zero-test but with more complicated update functions.  To avoid
confusion, we will stick to that vocabulary, and use \(c_i\) for counters and
\(r_i\) for registers.

\paragraph{Cost register automata.}\label{def:cra}  In this work, we only
consider cost register automata over \(\bbK \in \{\bbZi, \bbNi\}\) where the registers
are updated using expressions that use \(\min\) and ``\(+c\)'' for
\(c \in \bbK\).  A precise, formal definition of the model will only be needed for
\Cref{prop:nf}; to present the main constructions, we will simply rely on the
following more intuitive definition.

A \kcra \(\cC\) of dimension \(k\) is a DFA equipped with \(k\) registers
\(r_1, r_2, \ldots, r_k\) taking values in \(\bbK\).  The initial values of the registers
are specified by a vector in~\(\bbK^k\), and each transition further induces a
transformation of the form:
\[(\forall i \in [k])\quad r_i \leftarrow \min \{r_1 + m_{1, i},\enskip r_2 + m_{2, i},\enskip
  \ldots,\enskip r_k + m_{k, i},\enskip m_{k+1,i}\}\enspace,\]%
where each \(m_{i, j}\) is in \(\bbK\) (hence it can be \(\infty\), making the
subexpression irrelevant).  Each final state is paired with an output function
of the shape:
\[\min \{r_1 + m_1,\; r_2 + m_2,\; \ldots,\; r_k + m_k,\; m_{k+1}\}\enspace,\]
where again the \(m_i\)'s are in \(\bbK\).

Given a word \(w \in A^*\), the value of \(\cC\) on \(w\), written \(\cC(w)\), is \(\infty\) if
\(w\) is not accepted by the underlying DFA, and otherwise computed in the
obvious way: the registers are initialized, then updated along the (single) run
in the DFA, and the output is determined by the output function at the final
state.

The \kcra is said to be \emph{copyless} (\kccra) if all the update functions
satisfy, using the notations above, that for all \(i \in [k]\),
\(|\{j \mid m_{i, j} \neq \infty\}| \leq 1\); in words, for each \(i\), at most one of the
subexpressions ``\(r_i + m_{i, j}\)'' will evaluate to a non-\(\infty\) value: the value
of \(r_i\) impacts at most one register.

\paragraph{Vector addition systems with states and zero-tests.} The main
construction of this paper focuses on simulating counters with zero-tests.  The
precise formalism for our counter machines is a variant of vector addition
systems with states~(VASS) over \(\bbZ^k\), equipped with transitions that can
only be fired if a designated counter is zero.  For any \(k\), we define the
\emph{update alphabet} \(C_k\) as:
\[C_k = \bigcup_{i \in [k]} \{\linc_i, \ldec_i, \lchk_i\}\enspace,\]
the intended meaning being that \(\linc_i\) will increment the \(i\)-th counter,
\(\ldec_i\) will decrement it, and \(\lchk_i\) will check that it is zero.

A \zvz \(\cV\) of dimension~\(k\) is a DFA \((Q, C_k, \delta, q_0, F)\).  Consider a
|configuration| \(K = (q, \vec{c}) \in Q \times \bbZ^k\); writing
\((\vec{e_i}) \in\bbZ^k\) for the standard basis:
\begin{itemize}
\item If \(\delta(q, \linc_i) = q'\), then \(K\) can reach the configuration \((q', \vec{c} + \vec{e_i})\);
\item If \(\delta(q, \ldec_i) = q'\), then \(K\) can reach the configuration \((q', \vec{c} - \vec{e_i})\);
\item If \(\delta(q, \lchk_i) = q'\), then \(K\) can reach the configuration \((q', \vec{c})\) iff \(c_i = 0\).
\end{itemize}
We say that the \zvz reaches a state \(q\) if \((q_0, \vec{0})\) reaches, by a
sequence of configurations, \((q, \vec{c})\) for some \(\vec{c}\).  We write
\(L_{\cV, q} \subseteq (C_k)^*\) for the |reachability language| of \(q\), that is, the
language of updates along the runs reaching \(q\).

\begin{proposition}\label{prop:undec}
  The following problem is undecidable:\\
  \mbox{\qquad}\begin{tabular}{ll}
  \textbf{Given:}& A \zvz \(\cV\) and a state \(q\)\\
  \textbf{Question:} & Is \(L_{\cV, q}\) empty?
  \end{tabular}

  \noindent The problem stays undecidable even if \(|L_{\cV, q}| \leq 1\) is guaranteed.
\end{proposition}
\begin{proof}
  We define an extension of \zvz that can implement classical Minsky machines to
  streamline the reduction.  Define
  \(C'_k = C_k \cup \bigcup_{i\in [k]} \{\bar{\lchk}_i\}\).  A \(k\)_counter machine is an
  automaton over \(C'_k\), with the \zvz semantics, augmented with the property
  that a transition labeled \(\bar{\lchk}_i\) can only be taken if the \(i\)-th
  counter is \emph{non}zero.


  Minsky~\cite{minsky61} showed that the emptiness of reachability languages is
  undecidable for these machines---in particular, even if it is assumed that there
  is at most one run reaching the given state.  To show the same for \(\zvz\), we
  need only remove the transitions labeled \(\bar{\lchk}_i\), while preserving the
  reachability languages.  To do so, it suffices to replace them with the
  following gadget, where \(j\) is a new counter and some states are omitted:

  
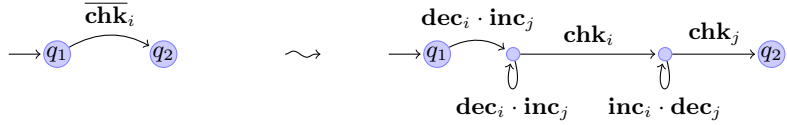
\begin{center}
    \begin{smallautomaton}[font=\footnotesize]
      \node[state, initial] (q1) {\(q_1\)};
      \node[state] at (1.4cm, 0cm) (q2) {\(q_2\)};
      \path[->] (q1) edge [above,bend left] node {\(\bar{\lchk}_i\)} (q2);

      \draw [->,decorate,decoration={snake,amplitude=.4mm}] (3cm, 0) -- +(0.5cm, 0);

      \begin{scope}[xshift=5cm]
        \node[state, initial] (q1) {\(q_1\)};
        \node[smallstate] at (3cm, 0cm) (qa) {};
        \node[state] at (4.4cm, 0cm) (q2) {\(q_2\)};
        \node[smallstate] at (1cm, 0cm) (q) {};

        \path[->] (q1) edge [bend left,above] node {\(\ldec_i\cdot\linc_j\)} (q)
                  (q) edge [loop below] node {\(\ldec_i\cdot\linc_j\)} (q)
                  (q) edge node [above] {\(\lchk_i\)} (qa)
                  (qa) edge [loop below] node {\(\linc_i\cdot\ldec_j\)} (qa)
                  (qa) edge node [above] {\(\lchk_j\)} (q2);
      \end{scope}
  \end{smallautomaton}
\end{center}
It is easily checked that upon reaching state \(q_2\), the \(i\)-th counter is
restored to its value in \(q_1\), the \(j\)-th is 0, and the state can only be reached
if the \(i\)-th counter were strictly positive.\qed
\end{proof}

\section{\ccra* and weighted automata}\label{sec:wa}

With the plethora of models computing functions from words to values in modern
literature, it is imperative to justify studying the seemingly artificial \ccra.
In this section, we provide a normal form that will demonstrate that these
machines are but deterministic weighted automata with a small dose of
nondeterminism.  In particular, all the problems we show to be undecidable in
\Cref{sec:app} turn out to be decidable for deterministic (or even
finitely-ambiguous) weighted automata; this gives credence to the assertion that
\nccra is one of the weakest models for which equivalence, for instance, is
undecidable.




In the following proposition, it is shown that any \kccra can be expressed as a
DFA making nondeterministic jumps into a \bbK-DetWA; graphically, every \kccra is
equivalent to:

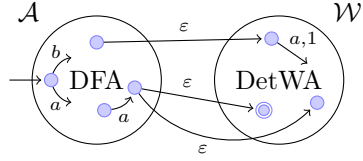
\begin{center}
  \begin{smallautomaton}\label{fig:nf}
    \node[circle, draw, minimum width=1.7cm] {DFA};
    \node[smallstate, initial] at (-0.6cm, 0) (ai) {};
    \node[smallstate] at (0, 0.5cm) (a1) {};
    \node[smallstate] at (0.1cm, -0.4cm) (a2) {};
    \node[smallstate] at (0.5cm, -0.1cm) (a3) {};

    \node at (-0.9cm,0.9cm) {\cA};

    \path[->] (a2) edge [bend right] node [below] {\scriptsize \(a\)} (a3)
              (ai) edge [bend left] node [above=.01cm,pos=0.2] {\scriptsize \(b\)} +(0.3cm, 0.3cm)
              (ai) edge [bend right] node [below=.01cm,pos=0.2] {\scriptsize \(a\)} +(0.3cm, -0.3cm);

    \begin{scope}[xshift=2.4cm]
      \node[circle, draw, minimum width=1.7cm] (b) {DetWA};
      \node[smallstate] at (-0.1cm, 0.55cm) (b1) {};
      \node[smallstate,accepting] at (-0.2cm, -0.4cm) (b2) {};
      \node[smallstate] at (0.5cm, -0.3cm) (b3) {};

      \path[->] (b1) edge node [above=0.1cm,sloped] (x) {} +(.5cm, -.35cm);
      \node at (x) {\scriptsize \(a,\!1\)};
      \node at (0.9cm,0.9cm) {\cW};
    \end{scope}

    \path[->] (a3) edge node [above,pos=0.4] {\scriptsize \(\eps\)} (b2)
              (a3) edge [bend right=50] node [below, pos=0.4] {\scriptsize \(\eps\)} (b3)
              (a1) edge node [above] {\scriptsize \(\eps\)} (b1);
    
  \end{smallautomaton}
\end{center}
\begin{proposition}\label{prop:nf}
  Let \(f\colon A^* \to \bbK\) be a \kccra.  There are a DFA
  \(\cA\) with state set \(Q\) and initial state \(q_0\), a \bbK_\detwa \cW with
  state set \(Q'\), and a function \(\eta\colon Q \to \pow(Q')\) such that:
  \[(\forall w \in A^*)\qquad f(w) = \min \{\cW^q(v) \mid w = uv \land q \in \eta(q_0.u)\}\enspace,\]
  where \(\cW^q\) is \cW with the initial state set to \(q\), and \(q_0.u\) is the
  state reached by reading \(u\) in \(\cA\).
\end{proposition}
\begin{proof}
  We first sketch the proof idea.  Consider a nondeterministic variant of a
  given \kccra \(\cC\) where updates of the form
  \(r_1 \leftarrow \min\{r_2, r_3\}\) become nondeterministic jumps between the updates
  \(r_1 \leftarrow r_2\) and \(r_1 \leftarrow r_3\).  The final value of this variant is set to be the
  minimum output of any run.  Then this variant has the same output value as the
  original CRA, by distributivity of addition over \(\min\).  We implement that
  strategy using a DFA \(\cA\) which, on resets (\(r_1 \leftarrow 0\)), starts a new run
  within a DetWA that follows the increments (\(r_1 \leftarrow r_1 + m\)) and movements
  (\(r_i \leftarrow r_1\)) of the register.

  We now formalize the definition of CRA.  A \kccra \(\cC\) of dimension \(k\) is a
  tuple \((\cA', \vec{\lambda}, \mu, \nu)\) where \(\cA' = (Q, A, \delta, q_0, F)\) is a DFA,
  \(\vec{\lambda} \in \bbK^{1\times k}\) is the initial value of the \(k\) registers,
  \(\nu\colon F \to \bbK^{(k+1)\times1}\) gives the output function for each final state,
  and \(\mu\colon Q \times A \to \bbK^{(k+1)\times(k+1)}\) provides the update
  functions.  To compare with the definition on page~\pageref{def:cra}, using
  the notation therein, \(\mu(q, a)\) is:
  \[
    \begin{pmatrix}
      m_{1, 1} & m_{1, 2} & \cdots & m_{1, k} & \infty\\
      m_{2, 1} & m_{2, 2} & \cdots & m_{2, k} & \infty\\
      \vdots & \vdots & \ddots & \vdots\\
      m_{k, 1} & m_{k, 2} & \cdots & m_{k, k} & \infty\\
      m_{k+1, 1} & m_{k+1, 2} & \cdots & m_{k+1, k} & 0
    \end{pmatrix}
  \]
  It can be readily checked that \((\vec{r}', 0) = (\vec{r}, 0)\cdot\mu(q, a)\) indeed
  satisfies, for all \(i \in [k]\) that:
  \[r_i' = \min\{r_1 + m_{1, i},\; r_2 + m_{2, i},\; \ldots,\; r_k+m_{k, i},\;
    m_{k+1,i}\}\enspace.\]%
  (Recalling that the multiplication is made in the
  semiring \((\bbK, \min, +)\).)
  Note that the \((k+1)\)-th component is a virtual register that will be
  maintained to \(0\).  Given an accepting run
  \((q_0, w_0, q_1)\cdot(q_1, w_1, q_2)\cdots(q_n,w_n, q_{n+1}) \in (Q \times A)^*\) in
  \(\cA'\), the output value is then defined as:
  \[\cC(w_0w_1\cdots w_n) = \vec{\lambda} \cdot \mu(q_0, w_0)\cdot\mu(q_1, w_1)\cdots \mu(q_n, w_n)\cdot
    \nu(q_{n+1})\enspace.\]

  Now that the precise definition of \kcra is settled, we present the
  construction.  We will assume that \(\vec{\lambda} \in \{0, \infty\}^k\) and that the updates are
  in one of two possible forms:
  \begin{itemize}
  \item \(r_i \leftarrow \min\{r_1 + m_1, r_2 + m_2 \ldots, r_k + m_k\}\), that is, no
    constant term appears;
  \item \(r_i \leftarrow 0\).
  \end{itemize}
  In symbols, this means that if \(\mu(q, a) = (m_{i,j})\), then for any \(i \in [k]\), either
  \(m_{k+1, i}\) is \(\infty\) or all \(m_{j, i}\), for \(j \in [k]\), are
  \(\infty\).  Any \(\kccra\) can be put under that form using standard techniques.
  
  The automaton \(\cA\) is the underlying automaton of \(\cC\), augmented with the
  information of \emph{which registers} were reset by the previous transition.
  More precisely, \(\cA = (Q \times \cP([k]), A, \delta_\cA, q_0', \emptyset)\) where
  \(q_0' = (q_0, \{i \mid \lambda_i = 0\})\); note that the final states are irrelevant.
  The transition function \(\delta_\cA\) is defined by:
  \[\delta_\cA((q, \cdot), a) = (\delta(q, a), E) \quad\text{where } E
    = \{i \mid \mu(q, a)_{k+1, i} = 0\}\enspace.\]
  
  The \bbK_\detwa \cW consists of \(k\) copies of \(\cC\), one for each register.
  Formally, \(\cW = (\cB,\vec{0}, \mu_\cW, \nu_\cW)\) with
  \(\cB= (Q\times[k], A, \delta_\cB, (q_0,1), F \times [k])\); here, the initial valuation is
  irrelevant.  We now define the transition function \(\delta_\cB\) and the weight
  function \(\mu_\cW\).  Let \((q, x)\) be a state of \(\cB\) and \(a \in A\).  By
  copylessness, there is at most one \(y\) such that \(\mu(q,a)_{x, y}\) is not
  \(\infty\).  If one such \(y\) exists, then:
  \begin{align*}
    \delta_\cB((q, x), a) & = (\delta(q, a), y)\\
    \mu_\cW((q, x), a) & = \mu(q,a)_{x, y}\enspace.
  \end{align*}
  The output function of \(\cW\) is then, for any \(q \in Q, i \in [k]\),
  \(\nu_\cW(q,i) = \nu(q)_i\).

  Finally, \(\eta\colon Q \times \cP([k]) \to \cP(Q \times [k])\) is defined as
  \(\eta(q, E) = \{(q, i) \mid i \in E\}\).

  Consider a word \(w \in A^*\), and a factorization \(w = uv\).  The word \(u\) reaches
  a state \(q\) in \(\cC\), and a state \((q, E)\) in \(\cA\).  The last transition
  taken in \(\cC\) reading \(u\) updated all the registers \(r_i, i \in E,\) with the
  value 0.  For each of these \(i\)'s, there will be a run over \(v\) in
  \(\cW\), starting at \((q, i)\), which follows the updates applied to \(r_i\).  This
  process thus simulates the nondeterministic variant of \(\cC\) described above,
  showing the Proposition.\qed
\end{proof}

\begin{corollary}\label{cor:lin}
  \(\kccra \subseteq \text{\bbK_\linwa}\).
\end{corollary}
\begin{proof}
  With the notations of \Cref{prop:nf}, let us see \(\cA, \eta,\) and \(\cW\) as a
  single \bbK_WA, where the weights in the \(\cA\) part are set to 0.  For any word
  \(w\), each run on \(w\) consists of a run over a prefix \(u\) within \(\cA\), and a
  run over the leftover suffix~\(v\) within \(\cW\) starting in some state
  \(q \in \eta(q_0.u)\).  Thus there are at most \(|w| \times |Q'|\) runs, hence the WA is
  linearly ambiguous.\qed
\end{proof}



As an application of this specific form, it is not hard to show that some
specific functions are not expressible using a \zccra.  Let \(\minblock\)
(\resp~\lastblock) be the function from \(\{a, \#\}^*\) to \(\bbN\) which, given
\(w = \#a^{n_1}\#a^{n_2}\#\cdots\#a^{n_k}\#\) returns
\(\min \{n_i\}_{i \in [k]}\) (\resp \(n_k\)):
\begin{proposition}\label{prop:unexp}
  The following functions are not expressible by a \zccra:
  \begin{itemize}
  \item \(c^i \cdot w \mapsto i + \minblock(w)\), with \(w \in \{a, \#\}^*\);
  \item \(u \cdot \$ \cdot v \mapsto \lastblock(u) + \lastblock(v)\), with \(u, v \in \{a, \#\}^*\).
  \end{itemize}
\end{proposition}
\begin{proof}[sketch]
  In both cases, one has to reason about when the nondeterministic jump, given
  by \(\eta\) in \Cref{prop:nf}, is made in the minimal run, bearing in mind that
  neither \(\minblock\) nor \(\lastblock\) are computable by a DetWA.

  For the first example, the jump has to be made at the beginning of the minimal
  block of \(a\)'s, after reading a \(\#\); thus the number of \(c\)'s cannot be taken
  into account.  For the second example, if the jump is made just before the
  last block of \(a\)'s in \(v\), then the value of the last block in \(u\) is
  disregarded.  If it is made just before the last block in \(u\), then the DetWA
  part has to compute \(\lastblock\) on \(v\), which is not possible.\qed
\end{proof}

\begin{remark}
  Note that the first function of \Cref{prop:unexp} is expressible by a LinWA,
  and the second by an unambiguous WA (i.e., at most one run per accepted word).
  Moreover, since \(\minblock\) is not expressible by an unambiguous WA but is by
  a \ccra (see the Introduction), the classes of functions expressed by the two
  models are incomparable.  We also remark that \Cref{prop:nf} and
  \Cref{cor:lin} hold for any semiring.
\end{remark}



\section{Simulation of \zvz* using \zccra*}

Let \(\cV\) be a \zvz and \(q\) a state of \(\cV\).  Recall that \(C_k\) is the update
alphabet of symbols \(\linc_i,\ldec_i,\) and \(\lchk_i\), and that
\(L_{\cV, q} \subseteq (C_k)^*\) is the reachability language of \(q\).  In this section, we
devise a simulation of \(\cV\) using \zccra in the following sense: Given a word
\(w \in (C_k)^*\), the \zccra will output 0 iff \(w \in L_{\cV, q}\).

Compared to the simulation by \nccra of the forthcoming \Cref{sec:nccra}, the
\bbZ case is quite straightforward, and reminiscent of the methodology
of~\cite{almagor-boker-kupferman11}; it however provides some intuition for the
construction for \bbN.

We present how the counter increments (\(\linc\)), decrements (\(\ldec\)), and
zero-tests (\(\lchk\)) are implemented for a single counter before showing how
multiple counters can be handled.  The automaton structure of the source \zvz,
with accepting state \(q\), can then be followed by the CRA while simulating the
counters.

\subsection{Simulation of a single counter}

Since we are working with a single counter, we drop the indices of the letters
in~\(C_1\).  A single counter \(c\) will be simulated by 3 registers: \(\rp\) and
\(\rn\), carrying the values of \(c\) and \(-c\), respectively, and \(\rchk\) which
shall be 0 if each time the letter \(\lchk\) was read, \(c\) was 0.  If at any time
\(\lchk\) was read while~\(c\) was nonzero, then \(\rchk\) will be strictly smaller
than 0.  This is implemented as follows:
\[\linc\colon\begin{update}
    \rp & \rp + 1\\
    \rn & \rn - 1\\
    \rchk & \rchk
  \end{update}\quad
  \ldec\colon\begin{update}
    \rp & \rp - 1\\
    \rn & \rn + 1\\
    \rchk & \rchk
  \end{update}\quad
  \lchk\colon\begin{update}
    \rp & 0\\
    \rn & 0\\
    \rchk & \min \{\rchk, \rp, \rn\}
  \end{update}
\]

\begin{observation}
  If \(\rchk\) becomes strictly smaller than 0, it will stay so after reading any
  word in \((C_1)^*\).
\end{observation}

\begin{observation}
  Assume \(\rp = \rn = \rchk = 0\).  After reading \(i\) letters \(\linc\) and
  \(j\)~letters \(\ldec\), in any order, then reading a final \(\lchk\), the new
  values of the registers satisfy:
  \begin{enumerate}
  \item If \(i = j\), then \(\rp = \rn = \rchk = 0\);
  \item Otherwise \(\rchk < 0\).
  \end{enumerate}
\end{observation}

This simulates the original counter in the following sense:
\begin{proposition}
  Let \(\cV\) be a \zvz of dimension \(1\) and \(q\) a state of \(\cV\). There is a
  \zccra \(\cC\) with \(\cC(w) \leq 0\) for any \(w\) and such that:
  \[(\forall w \in (C_1)^*)\quad w \in L_{\cV, q} \Leftrightarrow \cC(w) = 0\enspace.\]
\end{proposition}
\begin{proof}
  Let \(\cV = (Q, C_1, \delta, q_0, F)\) and \(q \in Q\).  The \zccra \(\cC\) with 3
  registers is defined as having \((Q, C_1, \delta, q_0, \{q\})\) as the automaton
  structure, and the updates are dictated by the letter being read, as
  above.  On state \(q\), \(\cC\) outputs~\(\rchk\).\qed
\end{proof}

\subsection{Simulation of multiple counters}

It is quite straightforward to combine multiple \(\rchk\) registers into one.
Indeed, if \(k\) counters are simulated using registers \(\rp_i, \rn_i,\) and
\(\rchk_i\), \(i \in [k]\), then at the end of the simulation, one can set:
\[\flag \leftarrow \min \{\rchk_1, \rchk_2, \ldots, \rchk_n\}\enspace,\]
so that \(\flag\) would be 0 iff the execution saw no illegal zero-tests;
\(\flag\) is negative otherwise.

\begin{proposition}\label{prop:zsim}
  Let \(\cV\) be a \zvz of dimension \(k\) and \(q\) a state of \(\cV\).  There is a
  \zccra \(\cC\) with \(\cC(w) \leq 0\) for any \(w\) and such that:
  \[(\forall w \in (C_k)^*)\qquad w \in L_{\cV, q} \Leftrightarrow \cC(w) = 0\enspace.\]
\end{proposition}

\begin{remark}\label{rmk:zreg}
  Here, we were mostly interested in having a specific output if the simulated
  execution was correct.  If we wanted, by contrast, to output one of the
  counters on correct executions, we would need one more idea; we present it
  here since it is similar to the techniques of the next section.

  Suppose that we wish to output the register \(r\) iff \(\flag\) is 0; recall that
  \(\flag\) may only be 0 or negative.  We will do so by repeatedly reading a new
  letter, and having \(r\) be the only possible \emph{even} output value, provided
  \(\flag\) is 0---no even value is produced if \(\flag\) is negative.

  We may assume that, by construction, \(\flag\) is even and \(r\) is a multiple
  of~4; we further assume that we have a register \(\rh\) that contains
  \emph{half} of \(r\)'s value.  We add the letter \(z\) to our alphabet, to be read
  at the end of the simulation; reading~\(z\) increases \(\rh\) by 2 and
  \(\flag\) by 4.  The output value is then set to:
  \[\min \{r + 1, \flag + 1, \rh\}\enspace.\]
  Write \(s\) for the value of \(r\) before reading the \(z\)'s, and \(f\) for the value
  of \(\flag\).  After reading \(i\) letters \(z\), the new values of the registers are:
  \[r = s,\quad\flag = f + 4\times i,\quad\rh = {s \over 2} + 2 \times i\enspace.\] For an even
  output to be produced, \(\rh\) has to be minimal.  If \(f\) is 0, this happens
  only when \(i = {s \over 4}\), and the output is then \(s\).  If \(f\) is negative,
  then \(\flag < \rh\) for \(i \leq {s \over 4}\) and \(r < \rh\) for larger values of
  \(i\); in that case, no even output value is produced.  This is illustrated in
  the following graphics, where \(s = 4\), and the left-hand side depicts
  the case \(f = 0\), while, in the right-hand side, \(f = -2\).
  \begin{center}
    \tikzfading[name=fade right, left color=transparent!0, right color=transparent!100]
    \tikzfading[name=fade left, left color=transparent!100, right color=transparent!0]
    \begin{tikzpicture}[scale=0.4]
      \clip (-1.7,-2) rectangle (23.5,8.5);
      \foreach \x in {0, ..., 9} {
        \draw node at (-1.5, \x) {\scriptsize\(\x\)};
      }

      \foreach \vflag/\vfpos/\vmin/\shift/\gl/\gr in {1/2/4/0/-0.5/11.5,-1/4/3/13cm/-1.5/10.5} {
        \begin{scope}[xshift=\shift]
          \clip (-1.5, -1.5) rectangle (11.5, 10.5);
          \draw[ystep=1.0,xstep=4.0,backstate!70,thin] (\gl,-1.5) grid (\gr,10.5);
          \foreach \x in {2, 6, 10} {
            \draw node at (\x, -0.5) {\(z\)};
          }

          \draw[densely dashdotted, semithick, red!60!black] (-.3, 2) -- ++(.3, 0) -- ++(10, 5);
          \draw node at (0.8, 3.2) {\(\rh\)};

          \draw[dashdotted, semithick, green!60!black] (-.3, 5) -- ++(.3, 0) -- ++(10, 0);
          \draw node at (1.5, 5.5) {\(r + 1\)};

          \draw[densely dotted,semithick, blue!60!black] (-.3, \vflag) -- ++(0.3, 0) -- ++(10, 10);
          \draw node at (\vfpos, 0.8) {\(\flag + 1\)};

          \draw[red, thick] (0, 6.2) node [above,fill=white] {min} -- (0, \vflag) node [inner sep=1,
          circle, fill=red] {};
          \draw[red, thick] (4, 6.2) node [above,fill=white] {min} -- (4, \vmin) node [inner sep=1,
          circle, fill=red] {};
          \ifnum \vmin=4
            \path[shorten <=1pt,<-] (4, 4) node[inner sep=.1cm,draw,circle] {} edge [bend right] (5,3);
            \node at (6, 2.7) {even};
          \fi
          \draw[red, thick] (8, 6.2) node [above,fill=white] {min} -- (8, 5) node [inner sep=1,
          circle, fill=red] {};
        \end{scope}
      }
      \draw (-0.2, 0) -- (30,0);
      \fill[white,path fading=fade left] (10.5,15) rectangle (11.5,-2);
      \fill[white,path fading=fade right] (11.5,15) rectangle (12.5,-2);
      \draw (11.5, -10) edge +(0, 30);
    \end{tikzpicture}
  \end{center}
\end{remark}

\section{Simulation of \zvz* using \nccra*}\label{sec:nccra}

Let \(\cV\) be a \zvz and \(q\) a state of \(\cV\).  In this section, we devise a
simulation of \(\cV\) using \nccra in the following sense: Given a word
\(w \in (C_k)^*\), the \nccra will output an \emph{even value} iff \(w \in L_{\cV, q}\).

Translating the strategy for \(\bbZ\) to the \bbN setting turns out to be a
nontrivial matter.  Indeed, one might expect that it would be enough to increase
the updates so that no negative number appears therein.  This would contribute a
linear blowup to the values, but does not seem to change the overall behavior.
However, the resets made while reading \(\lchk_i\) would have to be equal to that
blowup, and this would require copying.

The simulation will thus follow two phases.  First, one that corresponds to the
strategy for \(\bbZ\) with the updates tweaked to be positive; second, after
reading a \(\lchk_i\), a \emph{climb-back} phase that puts the registers back in a
manageable state (called ``ready'' later on).  For this latter phase, the \nccra
will read a word in \(\lcb_i^*\cdot\lchkcb_i\)---the letter \(\lcb\) standing for
\emph{climb-back}.  Further, combining the acceptance conditions of multiple
counters will also require some new letters; the alphabet of the automaton is
thus:
\[C'_k = C_k \cup \bigcup_{i \in [k]} \{\lcb_i, \lchkcb_i, z_i\}\enspace.\]

\subsection{Simulation of a single counter}

Again, since we are working with a single counter, we drop the indices of the
letters in~\(C'_1\).  A single counter in the \zvz will be simulated by 7
different registers, each with a simple intended meaning:
\begin{itemize}
\item \(\rp\) and \(\rn\) should respectively count the number of increments and
  decrements of the counter;
\item \(\rupd\) increases each time the counter is either incremented or
  decremented; it counts the number of |updates| to the counter;
\item The register \(\rhupd\) should be half of \(\rupd\);
\item \(\rchk\) will be a witness that the \(\lchk\) letter has always been read
  when the simulated counter was zero and that the climb-back phases were done
  correctly;
\item Finally, we will need two internal registers \(\rcb\) and \(\rdcb\), used
  solely in the climb-back phase.
\end{itemize}

\noindent To simplify the discussion, we give names to some register
configurations:
\begin{itemize}
\item They are \emph{ready} if \(\rp = \rn = \rchk = \rhupd = {1 \over 2} \times
  \rupd\);
\item They are \emph{to-climb} if \(\rp = \rn = 0\) and \(\rchk = \rhupd = {1 \over
    2} \times \rupd\);
\item They are \emph{dead} if \(\rchk < \rhupd\).
\end{itemize}
In the first two configurations, we also assume that the \(\rcb = \rdcb = 0\).

\paragraph{Goal of the construction.}  We will show that if the registers are
ready and we read an equal number of \(\linc\)'s and \(\ldec\)'s followed by a
\(\lchk\), then the registers become to-climb.  There is then a precise number
\(i\) such that reading \(\lcb^i\cdot\lchkcb\) will put the registers back in ready
mode.  Crucially, if the numbers of \(\linc\)'s and \(\ldec\)'s are not equal, or an
incorrect number of \(\lcb\)'s is read, then the registers become dead.

The updates are as follows, where the registers not shown are simply preserved.
As we saw in \Cref{rmk:zreg}, we will require that the values of the registers
be divisible by some values, hence rather than incrementing with 1, we increment
by a value \(e \in \bbN\) (for |Einheit|, unit) to be determined later.  Note that
these are indeed copyless updates.

\[\linc\colon\begin{update}[1]
    \rp & \rp + e\\
    \rupd & \rupd + e\\
    \rhupd & \rhupd + \frac{e}{2}\\
    \rchk & \rchk + \frac{e}{2}
  \end{update}\quad
  \ldec\colon\begin{update}[1]
    \rn & \rn + e\\
    \rupd & \rupd + e\\
    \rhupd & \rhupd + {e \over 2}\\
    \rchk & \rchk + {e \over 2}
  \end{update}\quad
  \lchk\colon\begin{update}[1]
    \rp & 0\\
    \rn & 0\\
    \rchk & \min \{\rchk, \rp, \rn\}
  \end{update}
\]

\[\lcb\colon\begin{update}[1]
    \rp & \rp + e\\
    \rn & \rn + e\\
    \rhupd & \rhupd + {e \over 2}\\
    \rchk & \rchk + {e \over 2}\\
    \rcb & \rcb + e\\
    \rdcb & \rcb + 2 \times e
  \end{update}\qquad
  \lchkcb\colon\begin{update}[1]
    \rcb & 0\\
    \rdcb & 0\\
    \rupd & \rdcb\\
    \rchk & \min\{\rchk, \rcb, \rupd\}\\
  \end{update}
\]

\begin{observation}
  If the registers are dead, they will stay so after reading any word in
  \((C'_1)^*\).
\end{observation}

\begin{lemma}
  Assume the registers are ready.  After reading \(i\) letters \(\linc\) and
  \(j\)~letters \(\ldec\), in any order, then reading a final \(\lchk\), the new
  values of the registers satisfy:
  \begin{enumerate}
  \item If \(i = j\), then they are to-climb;
  \item Otherwise, they are dead.
  \end{enumerate}
\end{lemma}
\begin{proof}
  Suppose \(\rp = \rn = \rchk = \rhupd = {1 \over 2} \times \rupd\), and let us name
  that value \(s\).  After reading \(i\) letters \(\linc\) and \(j\) letters
  \(\ldec\), the new values are:
  \[ \rp = s + e \times i,\quad
    \rn = s + e \times j,\quad
    \rchk = \rhupd = {1 \over 2} \times \rupd = s + e \times \frac{i+j}{2}\enspace.\]
  Now, if \(i = j\) then \(\rp = \rn = \rchk = \rhupd = {1 \over 2}\times \rupd\), thus
  reading \(\lchk\) will indeed make the registers to-climb.  Otherwise, one of
  \(\rp\) or \(\rn\) is smaller than \(\rchk\), and reading \(\lchk\) will make the
  registers dead.  \qed
\end{proof}

\begin{lemma}\label{lem:cb}
  Assume the registers are to-climb.  After reading \(\lcb^i\cdot\lchkcb\), the new
  values of the registers satisfy:
  \begin{enumerate}
  \item If \(i\) is equal to the starting value of \(\rchk\) multiplied by
    \({2 \over e}\), then they are ready;
  \item Otherwise, they are dead.
  \end{enumerate}
\end{lemma}
\begin{proof}
  Suppose \(\rp = \rn = 0\) and \(\rchk = \rhupd = {1 \over
    2} \times \rupd\); we name that latter value~\(s\).  After reading \(i\) letters
  \(\lcb\), the new values are:
  \[ \rp = \rn = \rcb = {1 \over 2} \times \rdcb = e \times i,\quad \rchk = \rhupd = s + e \times {i
      \over 2},\quad \rupd = 2 \times s\enspace.\] Now if
  \(i = {2 \times s \over e},\) then
  \(\rp = \rn = \rcb = {1 \over 2} \times \rdcb = \rchk = \rhupd = 2\times s\).  Reading
  \(\lchkcb\) thus makes the registers ready.  If \(i\) is smaller than
  \({2 \times s \over e}\) then \(\rcb < \rchk\); if it is greater, then
  \(\rupd < \rchk\): reading \(\lchkcb\) thus makes the registers dead.\qed
\end{proof}

\subsection{Simulation of multiple counters}

We just saw how to simulate a single counter in the sense that the registers are
not dead iff the input word describes a correct run (\ie, one in which \(\lchk\)
is only read if the counter is 0).  Let us now exhibit a method that combines
multiple such simulations, and outputs an even value iff none of the simulations
is dead.  To do so, we will repeatedly read new letters
\(z_1, z_2, \ldots, z_k\) at the very end of the execution, in a similar fashion as
\Cref{rmk:zreg}.

Let us suppose we have \(k\)~simulated counters, hence \(k\) sets of 7~registers.
For this phase, we will only use \(\rchk_i\), for each \(i\), but we will have |one|
more register in our \nccra, named \(\ravg\).  The purpose of \(\ravg\) is to hold
the average of all the~\(\rhupd_i\); this is easily achieved by adding to the
above updates:
\[\ravg \leftarrow \ravg + {e \over 2\times k}\enspace\]
whenever a \(\rhupd_i\) is incremented (always by \({e \over 2}\)).  Now for each
  \(i\), the new letter \(z_i\) will update the registers with:
\[
  \begin{update}
    \rchk_i & \rchk_i + {e \over 2}\\
    \ravg & \ravg + {e \over 2 \times k}
  \end{update}
\]
The output value of the \nccra is then set to
\begin{align}
  \min \{\ravg, \rchk_1 + 1, \rchk_2 + 1, \ldots, \rchk_k + 1\}\enspace. \label{eqn:out}  
\end{align}
We further assume
that \(e\) was chosen so that all the registers are even.

If \(\ravg\) was the average of the \(\rchk_i\)'s before reading the
\(z_i\)'s---and this only happens if none of the register set was dead---it will stay
so reading \(z_i\)'s.  Consequently, there is a number of each letter \(z_i\) that
can be read so that all the \(\rchk_i\)'s are equal, making \(\ravg\) the output
value of the \nccra.

If \(\ravg\) was greater than the average of the \(\rchk_i\)'s---implying that at
least one set of registers was dead---then \(\ravg\) will never be the output of the
\ccra after reading \(z_i\)'s.

\begin{theorem}[Simulation]\label{thm:sim}
  Let \(\cV\) be a \zvz of dimension \(k\) and \(q\) a state of \(\cV\).  Write
  \(h\colon (C'_k)^* \to (C_k)^*\) for the function that erases the letters
  \(\lcb_i, \lchkcb_i,\) and \(z_i\).  There is an \nccra \(\cC\) such that for all
  \(w \in (C_k)^*\):
  \[w \in L_{\cV, q} \Leftrightarrow (\exists!w' \in h^{-1}(w))[\cC(w') \text{ is even}]\enspace.\]
\end{theorem}
\begin{proof}
  The only detail left to deal with is the uniqueness of the \(w'\).  We can
  certainly make sure that \(\cC\) outputs a value iff the input is of the
  form:
  \[(\linc_i + \ldec_i + \lchk_i\cdot\lcb_i^*\cdot\lchkcb_i)_i^* \cdot (z_i)_i^*\enspace,\]
  but even if the first half (without the \(z_i\)'s) is indeed unique, as per
  \Cref{lem:cb}, the~\(z_i\)'s need not be so.  To preserve uniqueness, this latter
  part is replaced by:
  \[\bigcup_{j \in [k]} \enspace \prod_{%
      \substack{%
         i = 1, \ldots, k\\
        i \neq j
      }} (z_i)^*\enspace.\]

  This serves two purposes: first, the order on the \(z_i\)'s is fixed; second,
  one of the \(z_j\) will \emph{not} be used, hence the condition that all the
  \(\rchk_i\) be equal will only be satisfied when they are all valued
  \(\rchk_j\).  Naturally, such a \(j\) exists, it is simply the index of a maximal
  \(\rchk_i\), making %
  {\ooalign{\hfil\(\prod\)\hfil\crcr\smash{\raisebox{-1em}{\(\scriptstyle i = 1, \ldots, k\)}}}} %
  \((z_i)^{\rchk_j-\rchk_i}\) the only possible suffix leading to an even value.\qed
\end{proof}

\section{Applications}\label{sec:app}

We draw a number of undecidability results as consequences of these simulations.

\begin{theorem}[Equivalence]
  The following problem is undecidable:\\
  \mbox{\qquad}\begin{tabular}{ll}
  \textbf{Given:}& Two \nccra \(\cC\) and \(\cC'\) over \(A^*\)\\
  \textbf{Question:} & \((\forall w \in A^*)[\cC(w) = \cC'(w)]\)
  \end{tabular}
\end{theorem}
\begin{proof}
  Let \(\cV\) be a \zvz and \(q\) a state of \(\cV\), and consider the
  \nccra~\(\cC\) that simulates \(L_{\cV, q}\).  We reduce deciding if that language
  is empty (which is undecidable by \Cref{prop:undec}) to the problem at hand.
  \Cref{eqn:out}, defining the output of \(\cC\), is such that \(\ravg\) is the
  minimum iff the execution was correct.  Thus replacing this output function
  by:
  \[\min \{\rchk_1 + 1, \rchk_2 + 1, \ldots, \rchk_k + 1\}\]
  changes the output value of a word iff it was a correct run.  Calling \(\cC'\)
  this modified version, it holds that \((\forall w \in A^*)[\cC(w) = \cC'(w)]\) iff
  \(L_{\cV, q} = \emptyset\).\qed
\end{proof}

Clearly, it is undecidable whether the image of an \nccra is always odd.
Further, that image may be nonsemilinear (see the following proof), and:
\begin{theorem}[Semilinearity]\label{thm:sl}
  The following problem is undecidable:\\
  \mbox{\qquad}\begin{tabular}{ll}
  \textbf{Given:}& An \nccra \(\cC\) over \(A^*\)\\
  \textbf{Question:} & Is \(\cC(A^*)\) semilinear, i.e., an eventually periodic set?
  \end{tabular}
\end{theorem}
\begin{proof}
  We provide an independent construction which bears some similarities to the
  ``climb-back'' method.  It doubles a register \(r\) in the following sense:
  if~\(r\) is a register with starting value \(s\), then reading
  \(\linc^{s/2}\cdot\lchk\) \emph{doubles} the value of~\(r\); if any other number of
  \(\linc\)'s is read (which happens in particular when \(s\) is odd), the new value
  of \(r\) will be some odd number.

  Consider a register \(r\) with initial value \(s\), and suppose we have an
  additional register \(r'\) holding \(2 \times s\).  We introduce two new registers,
  \(\rcb\) and \(\rdcb\) initialized with~0.  Upon reading a word
  \(\lcb^i \cdot \lchkcb\), we apply the updates:
  \[\lcb\colon \begin{update}
      r & r + 2\\
      \rcb & \rcb + 4\\
      \rdcb & \rdcb + 8
    \end{update}\qquad
    \lchkcb\colon
    \begin{update}
      \rcb & 0\\
      \rdcb & 0\\
      r' & \rdcb\\
      r & \min\{r, r' + 1, \rcb + 1\}
    \end{update}
  \]
  After reading \(\lcb^i\), it holds that
  \(r = s + 2 \times i,\; \rcb = 4\times i,\) and \(\rdcb = 8 \times i\).

  If \(i = {s \over 2}\), then \(r = \rcb = r' = 2\times s\), hence after reading
  \(\lchkcb\), we have indeed \(r = {r' \over 2} = 2 \times s\), and the extra registers
  are reset: we are back to our starting hypothesis.

  If \(i \neq {s \over 2}\), then either \(r' < r\) (when \(i > {s \over 2}\)) or
  \(\rcb < r\) (when \(i < {s \over 2}\)).  In both cases, after reading
  \(\lchkcb\), \(r\) becomes odd, and will stay so after reading any other word.

  As a side note, consider the \nccra with the above updates and \(r\) initialized
  to 2, that reads words in \((\lcb^*\cdot\lchkcb)^*\).  Then the only even outputs of
  this machine are the powers of two, a nonsemilinear set.

  This concludes the construction, and we now present the reduction.

  Let \(\cV\) be a \zvz and \(q\) a state of \(\cV\), and consider the \nccra
  \(\cC\) that simulates \(L_{\cV, q}\).  We assume that
  \(|L_{\cV, q}| \leq 1\), and again reduce deciding \(L_{\cV, q} = \emptyset\) to the problem
  at hand.

  First we note that we may assume that \(\cC\) outputs all the odd numbers, for
  instance by adding a letter \(\ell\) and, upon reading \(\ell^n\), outputting \(2\times n +
  1\).  Also recall that if \(L_{\cV, q}\) is nonempty, then there is a unique \(w\)
  such that \(\cC(w)\) is even.

  We now modify \(\cC\) into \(\cC'\) to incorporate the above machinery.  We simply
  store in a new register \(r\) the output value of \(\cC\), and proceed by reading
  words of the form \(\lcb^i\cdot\lchkcb\) with the updates as above.  If
  \(L_{\cV, q} = \emptyset\), then \(\cC'((C'_k)^*)\) is all the odd numbers, a semilinear
  set.  Otherwise, there is one (and only one) even value \(s\) in the image of
  \(\cC\), and it holds that:
  \[\cC'((C'_k)^*) = (2\bbN + 1) \cup \{2^i \times s \mid i \geq 0\}\enspace,\]
  a nonsemilinear set.\qed
\end{proof}

\begin{theorem}[Upperboundedness]\label{thm:ub}
  The following problem is undecidable:\\
  \mbox{\qquad}\begin{tabular}{ll}
  \textbf{Given:}& A \(\bbZi\)-WA \(\cA\) over \(A^*\)\\
  \textbf{Question:} & \((\exists c \in \bbZ)(\forall w \in A^*)[\cA(w) \leq c]\)
  \end{tabular}
\end{theorem}
\begin{proof}
  Let \(\cV\) be a \zvz and \(q\) a state of \(\cV\), and consider the \zccra
  \(\cC\) that simulates \(L_{\cV, q}\).  Relying on \Cref{prop:nf}, let \(\cW\) be a
  \(\bbZi\)-WA equivalent to~\(\cC\).  Tweak \(\cW\) to output the same as
  \(\cC\) \emph{plus one}, hence \(\cW(w)\) is 1 iff \(w \in L_{\cV, q}\).  Now let
  \(\cW'\) be \(\cW\) with an added letter \(\#\) that jumps from the final states of
  \(\cW\) to its initial state; formally, let
  \(\cW = (\cA, \lambda, \mu, \nu)\) with \(\cA = (Q, A, \delta, q_0, F)\), then
  \(\cW'\) is \((\cA', \lambda, \mu', \nu)\) where
  \(\cA' = (Q, A \uplus \{\#\}, \delta \cup \{(q, \#, q_0) \mid q \in F\}, q_0, F)\), and
  \(\mu'\) agrees with \(\mu\) on \(\delta\) and is extended by \(\mu(q, \#, q_0) = \nu(q) + \lambda\).

  In essence, \(\cW'\) is iterating \(\cW\):
  \[\cW'(w_1\#w_2\#\cdots\#w_k) = \sum_{i \in [k]} \cW(w_i)\enspace.\]
  From this, we see that if \(\cW\) is always negative or zero, \(\cW'\) is bounded,
  otherwise, if \(\cW(w) = 1\), then \(\cW'((w\#)^c\cdot w) = c + 1\), hence \(\cW'\) is
  unbounded.\qed
\end{proof}

\section{Conclusion}

Deceptively powerful, copyless cost register automata with increments and
\(\min\) operations were shown to be able to simulate and check runs of counter
machines.  The constructions show that the repeated use of \(\min\) enables
behaviors that \emph{appear} outside the scope of \emph{copylessness}, e.g., an
\nccra can double the value of a register (or, more precisely, can attempt to do
so while knowing when it failed).  As a main consequence, equivalence of
\(\nccra\) is undecidable.

We wish to highlight two open questions.  First, \Cref{thm:ub} comes short of
telling us anything about the decidability of upper-boundedness for \zccra (the
same being decidable for \nccra and \bbN-WA in general~\cite{hashiguchi82}).
Note that it cannot be decided whether a \(\zccra\) is upper-bounded by a
\emph{given} constant (from \Cref{prop:zsim}).

Second, the normal form of \Cref{prop:nf} hints to the possibility that linearly
ambiguous WA can be put into a similar form.  More precisely, it seems that any
such WA can be decomposed into two \emph{unambiguous} WA, the first one making
nondeterministic jumps into the second.  Does this hold?

\paragraph{Acknowledgments.}  We would like to thank Ismaël Jecker, Andreas
Krebs, Mahsa Shirmohammadi, and James Worrell for stimulating discussions.

\bibliographystyle{plainurl}
\bibliography{bib}

\begin{thebibliography}{1}

\bibitem{almagor-boker-kupferman11}
Shaull Almagor, Udi Boker, and Orna Kupferman.
\newblock What's decidable about weighted automata?
\newblock In {\em {ATVA} 2011}, pages 482--491, 2011.
\newblock \href {http://dx.doi.org/10.1007/978-3-642-24372-1_37}
  {\path{doi:10.1007/978-3-642-24372-1_37}}.

\bibitem{alur-et-al13}
Rajeev Alur, Loris D'Antoni, Jyotirmoy~V. Deshmukh, Mukund Raghothaman, and
  Yifei Yuan.
\newblock Regular functions and cost register automata.
\newblock In {\em {LICS} 2013}, pages 13--22, 2013.
\newblock \href {http://dx.doi.org/10.1109/LICS.2013.65}
  {\path{doi:10.1109/LICS.2013.65}}.

\bibitem{cadilhac-finkel-mckenzie13}
Micha{\"{e}}l Cadilhac, Alain Finkel, and Pierre McKenzie.
\newblock Unambiguous constrained automata.
\newblock {\em Int. J. Found. Comput. Sci.}, 24(7):1099--1116, 2013.
\newblock \href {http://dx.doi.org/10.1142/S0129054113400339}
  {\path{doi:10.1142/S0129054113400339}}.

\bibitem{gaubert-katz04}
St{\'{e}}phane Gaubert and Ricardo Katz.
\newblock Rational semimodules over the max-plus semiring and geometric
  approach to discrete event systems.
\newblock {\em Kybernetika}, 40(2):153--180, 2004.

\bibitem{hashiguchi82}
Kosaburo Hashiguchi.
\newblock Limitedness theorem on finite automata with distance functions.
\newblock {\em Journal of Computer and System Sciences}, 24(2):233 -- 244,
  1982.
\newblock \href
  {http://dx.doi.org/https://doi.org/10.1016/0022-0000(82)90051-4}
  {\path{doi:https://doi.org/10.1016/0022-0000(82)90051-4}}.

\bibitem{mazowiecki-riveros16}
Filip Mazowiecki and Cristian Riveros.
\newblock Copyless cost-register automata: Structure, expressiveness, and
  closure properties.
\newblock In {\em {STACS} 2016}, pages 53:1--53:13, 2016.
\newblock \href {http://dx.doi.org/10.4230/LIPIcs.STACS.2016.53}
  {\path{doi:10.4230/LIPIcs.STACS.2016.53}}.

\bibitem{minsky61}
Marvin~L. Minsky.
\newblock Recursive unsolvability of {Post}'s problem of ``tag'' and other
  topics in theory of {Turing} machines.
\newblock {\em Annals of Mathematics}, 74(3):pp. 437--455, 1961.

\bibitem{presburger27}
Mojzesz Presburger.
\newblock {\"U}ber de vollst{\"a}ndigkeit eines gewissen systems der arithmetik
  ganzer zahlen, in welchen, die addition als einzige operation hervortritt.
\newblock In {\em Comptes Rendus du Premier Congr{\`e}s des Math{\'e}maticiens
  des Pays Slaves}, pages 92--101, Warsaw, 1927.

\end{thebibliography}

\end{document}